\documentclass[sigconf]{acmart}
\renewcommand\footnotetextcopyrightpermission[1]{} 
\usepackage{booktabs} 

\usepackage{threeparttable}
\usepackage{algpseudocode,amssymb,amsfonts,color,xspace}

\usepackage{subfiles}

\usepackage{algorithm2e}
\usepackage{mathtools}

\usepackage{siunitx}
\usepackage{tikz} 
\usepackage{pgfplots}

\pgfplotsset{compat=newest} 
\usepgfplotslibrary{units} 

\newcommand{\GF}[1]{\ensuremath{\mathbb{F}_{#1}}\xspace}
\newcommand{\Gal}[1]{\ensuremath{\mathrm{Gal}(#1)\xspace}}
\newcommand{\Orb}[1]{\ensuremath{\mathrm{Orb}_{#1}\xspace}}

\newcommand{\rd}[1]{\ensuremath{\lfloor \lg{#1} \rfloor }\xspace}
\newcommand{\binrd}[1]{\ensuremath{ 2^{\lfloor \lg{#1} \rfloor }}\xspace}
\newcommand{\binru}[1]{\ensuremath{ 2^{\lceil \lg{#1} \rceil }}\xspace}

\newcommand{\nexi}[1]{\ensuremath{\phi^{\circ #1}}\xspace}

\newcommand{\ifft}{\texttt{ifft}\xspace}

\newcommand{\AFT}{\ensuremath{\textit{AFT}_n\xspace}}
\newcommand{\FAFT}[1]{\ensuremath{\textit{FAFT}_{#1}\xspace}}

\newcommand{\lchbtfy}{\ensuremath{\texttt{AFFT}}\xspace}
\newcommand{\FAFFT}{\ensuremath{\texttt{FAFFT}}\xspace}

\newcommand{\novelpoly}{\emph{novel polynomial basis}\xspace}

\newcommand{\bvec}[1]{\ensuremath{\boldsymbol{#1}}\xspace}

\newtheorem{theorem}{Theorem}[section]

\theoremstyle{definition}
\newtheorem{definition}[theorem]{Definition}

\newif\ifsubmission
\submissionfalse

\begin{document}
\title{Frobenius Additive Fast Fourier Transform}

\author{Wen{-}Ding Li}
\affiliation{
  \institution{Research Center of Information Technology and Innovation, Academia Sinica, Taiwan}
}
\email{thekev@crypto.tw}

\author{Ming-Shing~Chen}
\affiliation{
  \institution {Research Center of Information Technology and Innovation, Academia Sinica, Taiwan}
}
\email{mschen@crypto.tw}

\author{Po-Chun Kuo}
\affiliation{
  \institution {Department of Electrical Engineering, National Taiwan University, Taiwan}
}
\email{kbj@crypto.tw}

\author{Chen-Mou~Cheng}
\affiliation{
   \institution{Department of Electrical Engineering, National Taiwan University, Taiwan}
}
\email{doug@crypto.tw}

\author{Bo-Yin~Yang}
\affiliation{
  \institution{Institute of Information Science, Academia Sinica, Taiwan}
}
\email{by@crypto.tw}


\begin{abstract}
  In ISSAC 2017, van der Hoeven and Larrieu showed that evaluating a
  polynomial $P\in\GF q[x]$ of degree $<n$ at all $n$-th roots of
  unity in \GF{q^d} can essentially be computed $d$-time faster than
  evaluating $Q\in\GF{q^d}[x]$ at all these roots, assuming $\GF{q^d}$
  contains a primitive $n$-th root of unity~\cite{ffft}.
  Termed the Frobenius FFT, this discovery has a profound impact on
  polynomial multiplication, especially for multiplying binary
  polynomials, which finds ample application in coding theory and
  cryptography.
  In this paper, we show that the theory of Frobenius FFT beautifully
  generalizes to a class of additive FFT developed by Cantor and
  Gao-Mateer~\cite{Cantor:1989,gao2010additive}.
  Furthermore, we demonstrate the power of Frobenius additive FFT for
  $q=2$: to multiply two binary polynomials whose product is of degree
  $<256$, the new technique requires only 29,005 bit operations, while
  the best result previously reported was 33,397.
  To the best of our knowledge, this is the first time that FFT-based
  multiplication outperforms Karatsuba and the like at such a low
  degree in terms of bit-operation count.
\end{abstract}

%
%



\maketitle

\section{Introduction}

Let \GF{q^d} be the finite field of $q^d$ elements, and let
$\xi\in\GF{q^d}$ be a primitive $n$-th root of unity.
The (discrete) Fourier transform of a polynomial $P\in\GF{q^d}[x]$
with degree $<n$ is $(P(1),P(\xi),P(\xi^2),\ldots,P(\xi^{n-1}))$,
namely, evaluating $P$ at all $n$-th roots of unity.
How to efficiently compute the Fourier transform not only is an
important problem in its own right but also finds a wide variety of
applications.
As a result, there is a long line of research aiming to find what is
termed ``fast'' Fourier tranform, or FFT for short, for various
situations.

Arguably, one of the most important applications of FFT is fast
polynomial multiplication.
In particular, the case of $q=2$ has received a lot of attention from
the research communities due to its wide-ranging application, e.g., in
coding theory and cryptography.
Here we obviously need to go to an appropriate extension field
\GF{2^d} in order to obtain a primitive $n$-th root of unity for any
meaningful $n$, and in this case, it is well known that one can use
the Kronecker method to efficiently compute binary polynomial
multiplication~\cite{Cantor:1989,DBLP:conf/issac/HarveyHL16}.
Such FFT-based techniques have better asymptotic complexity compared
with school-book and Karatsuba algorithms.
However, it is conventional wisdom that FFT is not suitable for
polynomial multiplication of small degrees because of the large hidden
constant in the big-$\mathcal{O}$ notation~\cite{Fan::Haining::Hasan::Anwar:2015}.

We recall that the Frobenius map $x\mapsto x^q$ fixes \GF q in any of
its extension field \GF{q^d}, and hence
$\forall P\in\GF q[x],\forall a\in\GF{q^d},P(\phi(a))=\phi(P(a))$.
In ISSAC 2017, van der Hoeven and Larrieu showed how to use Frobenius
map to speed up the Fourier transform of $P\in\GF q[x]$ essentially by
a factor of $d$ over $Q\in\GF{q^d}[x]$ and hence avoid the
factor-of-two loss as in the Kronecker method~\cite{ffft}.
However, the Frobenius FFT is complicated, especially when the
Cooley-Tukey algorithm is used for a (highly) composite $n$.
One of the reasons behind might be that the Galois group of \GF{q^n}
over \GF q is generated by the Frobenius map and isomorphic to a
cyclic subgroup of the \emph{multiplicative} group of units of
$\mathbb Z/n\mathbb Z$, whereas the Cooley-Tukey algorithms works by
decomposing the \emph{additive} group $\mathbb Z/n\mathbb Z$.
The complicated interplay between these two group structures can bring
a lot of headaches to implementers.

In his seminal work, Cantor showed how to evaluate a polynomial in
some additive subgroups of a tower of Artin-Schreier extensions of a
finite field and gave an $O(n\lg^{\log_3 2}(n))$ FFT algorithm based
on polynomial division~\cite{Cantor:1989}.
An Artin-Schreier extension of a finite field \GF q of characteristic
$p$ is a degree-$p$ Galois extension of \GF q.
In this paper, we restrict our discussion to the case of $p=2$, but
most of the results can be extended to the case of general $p$.
Based on Cantor's construction, Gao and Mateer gave a
Cooley-Tukey-style algorithm whose complexity is $O(n\lg(n)\lg\lg(n))$
when $d$ is a power of two~\cite{gao2010additive}, using which
Chen~\emph{et al.}~ achieved competitive performance compared with other state of the art of binary polynomial
multiplication~\cite{DBLP:journals/corr/abs-1708-09746}.
As will become clear later in this paper, the theory of Frobenius FFT
beautifully generalizes to additive FFT developed by Cantor and
Gao-Mateer because the group that FFT works on comes from the same
Frobenius map.

Frobenius additive FFT is not only interesting in its own right but
can be useful in a variety of applications.
In particular, many techniques to reduce the number of bit operations
(AND and XOR) of binary polynomial multiplications of small degrees
were proposed in the
literature~\cite{DBLP:conf/crypto/Bernstein09,DBLP:journals/jce/CenkH15,DBLP:journals/tc/CenkNH13,DBLP:journals/tc/CenkHN14,DBLP:conf/sacrypt/GathenS05}.
Although the number of bit operation is not an accurate performance
predictor in modern CPU, it is still a useful metric for digital
circuit design or ``bitslice'' software technique in embedded device.
However, so far most of the techniques for small degrees were based on
Karatsuba algorithm or its generalization to $n$-way split.
By applying Frobenius additive Fourier transform instead of Kronecker
method, we show that we can break the record for the number of bit
operations even at the polynomial size 231.
To the best of our knowledge, it is the first time FFT-based method is
shown to be competitive in such small degrees.
We also implement a code generator to output procedures for
multiplying
two polynomials, publicly available at
\begin{center} 
  \url{https://github.com/fast-crypto-lab/Frobenius_AFFT}
\end{center}

The rest of this paper is organized as follows.
In Section~\ref{sec:preliminaries}, we will review the relevant
background information.
In Section~\ref{sec:faft}, we will define the Frobenius additive
Fourier transform and show some of its important properties.
In Section~\ref{sec:application}, we conclude by showing how we apply
Frobenius additive FFT to binary polynomial multiplication and achieve
a new record.

\section{Preliminaries}
\label{sec:preliminaries}
\subsection{Basis of finite field}
\label{sec:gf-ver3}
Let \GF{2^d} denote an binary extension field, and let
\[ \bvec{v}_d = (v_0,v_1,\ldots,v_{d-1}). \]
We call $\bvec{v}_d$ a basis for \GF{2^d} if $v_0,v_1,\ldots,v_{d-1}$
are linearly independent over \GF 2.
Throughout this paper, we often represent an element $\omega_i$ of a
binary extension field as
\[ \omega_i = i_0 v_0 + i_1 v_1 + \cdots + i_{m-1} v_{m-1}, \] where
$i = i_0 + 2i_1 + 2^2i_2 +\cdots + 2^{m-1}i_{m-1}$,
$i_j\in\{0,1\}\forall 0\leq j<m$, with the basis elements
$v_0,v_1,\ldots,v_{m-1}$ inferred from the context.

In his seminal work, Cantor presented a sequence of explicit and
computationally useful bases for binary extension
fields~\cite{Cantor:1989}.
\begin{definition}
\label{Cantor_basis}
  Given a sequence $u_0,u_1,u_2,\ldots$ of elements from the algebraic
  closure of \GF 2 satisfying
  \[ u_i^2 + u_i = (u_0u_1\cdots u_{i-1}) + [\text{a sum of monomials
      of lower degrees}], \] where each ``monomial of a lower degree''
  has the form $u_0^{j_0} u_1 ^{j_1} \cdots u_{i-1}^{j_{i-1}}$ such
  that $\forall 0\leq k<i,j_k\in\{0,1\}$ and $\exists k,j_k=0$.
  Then a Cantor basis $\bvec{v}_d = (v_0,\ldots,v_{d-1})$ for \GF{2^d}
  is defined as \[ v_i = u_0^{i_0}u_1^{i_1}\cdots u_{k-1}^{i_{k-1}} \]
  where $i = i_0 + 2i_1 + \cdots + 2^{k-1}i^{k-1}$ and $d = 2^k$.
\end{definition}

If we fix $\GF{2^{2^k}}=\GF 2(u_0,u_1,\ldots,u_{k-1})$ for
$k=1,2,\ldots$, then with Cantor's construction, we arrive at a tower
of Artin-Schreier extensions of \GF 2.
For example, the following tower of extension fields of \GF 2 are one
such construction:
\[
\begin{array}{rcl}
  \GF{4} & := & \GF{2}[u_0]/(u_0^2 + u_0 + 1),\\
  \GF{16} & := & \GF{4}[u_1]/(u_1^2+u_1+u_0), \\
  \GF{256} & := & \GF{16}[u_2]/(u_2^2+u_2+u_1u_0), \\
  \GF{65536} & := & \GF{256}[u_3]/(u_3^2+u_3+u_2u_1u_0),\\
         & \vdots &
\end{array}
\]
In this case, the Cantor basis for, e.g., \GF{65536} is
\[ \bvec{v}_{16}=(1,u_0,u_1,u_0u_1,
  u_2,u_0u_2,u_1u_2,\ldots,u_0u_1u_2u_3). \]
In this paper, we will focus on additive Fourier transform with
respect to Cantor bases.

\subsection{Finite field arithmetic}
We will use the bit complexity model for finite field arithmetic
unless stated otherwise.
We use $M_q(d)$ to denote the complexity of multiplication of
polynomials of degree $<d$ over \GF q.
Currently, the best known bound for $M_q(n)$ is
\[ M_q(d) = \mathcal{O}(d \log q \log(d \log q) 8^{\log^*(d \log q)}), \] where
$\log^*(\cdot)$ is the iterated logarithm
function~\cite{DBLP:journals/jacm/HarveyHL17}.
It is conventional to assume that $M_q(d)/d$ is an increasing function
of $d$.
We will denote $M(d)$ as the bit complexity to multiply two elements in $\GF{2^d}$ represented in Cantor basis.
Since we can use modular decomposition technique\cite{composition}
\cite{isomorphism} to convert $\GF{2^d}$ to $\GF{2}[x]$ and then perform polynomial multiplication with $\mathcal{O}(d \lg d)$.
So $M(d) = \mathcal{O}(M_2(d))$.
We also assume that $\frac{M(d)}{d}$ is an increasing function in $d$ for Cantor bases.
We use $A(d)$ to denote the complexity of addition for two elements in $\GF{2^d}$.
As usual, the complexity of adding two elements in $\GF{2^d}$ is as
$O(d)$.
Note that in some case, Cantor's construction allows more efficient multiplication.
For example, given
$\alpha, \beta \in \GF{2^{2^k}}:=\GF{2^{2^{k-1}}}[u_{k-1}]/(u_{k-1}^2
+ u_{k-1} + \zeta)$, if $\alpha$ happens to be in the (proper)
subfield $\GF{2^{2^{k-1}}}$, then multiplication of $\alpha$ and
$\beta$ can be computed using only two multiplications in
$\GF{2^{2^{k-1}}}$.
The cost of multiplication become $2M(2^{k-1})$ instead of $M(2^k)$.
%
%
%
As we shall see, we often multiply elements from different extension
fields of \GF 2, so Cantor's trick plays an important role in reducing
bit complexity.

\subsection{Additive Fourier Transform}
Let $\bvec{v}_d = (v_0, v_1, v_2, ..., v_{d-1})$ 
be a basis of $\GF{2^d}$. Let $n=2^m$ and $m \leq d$.
Now consider a polynomial $P \in \GF{2^d}[x]_{< n}$, where 
$$
\GF{2^d}[x]_{<n} := \{P \in \GF{2^d}[x]: \deg (P) < n\} 
$$
We will define the additive Fourier transform $\AFT(P)$ with respect
to a basis $\bvec{v}_d$ to be
$$
\AFT(P) = 
\big( P(\omega_0), P(\omega_1),P(\omega_2),...,P(\omega_{n-1}) \big)
$$
Recall that $\omega_i = \sum_{j=0}^{m-1} i_j \cdot v_j$, $i = \sum_{j=0}^{m-1} i_j \cdot 2^j$ and $i_j \in \{0,1\}$

\subsection{Subspace polynomial}

Consider a basis $\bvec{v} = (v_i)_{i=0}^{d-1}$ and all $v_i \in \GF{2^d}$.
Let $$W_k:=\mathrm{span}\{v_0,v_1,\ldots, v_{k-1}\}
=\{\sum_{j \in S} v_j | S \subseteq \{0,1,2,...,k-1\} \}$$

denote an $k$-dimensional subspace in $\GF{2^d}$, where $k\leq d$. These $W_k$ satisfies 
$$
\{0\} = W_0 \subset W_1 \subset \cdots \subset W_d = \GF{2^d} 
$$ and form a sequence of subspaces. We  define $W_{0} = \{0\}$ for convenience later.

\begin{definition}\label{def:si}Given a subspace $W_k$ of $\GF{2^d}$, the subspace polynomial is defined as
$$
 s_k(x) := \prod_{a \in W_k} (x-a) \enspace.
$$
\end{definition}





\begin{lemma}
$s_k(x)$ is a linearized polynomial:
$$s_k(x+y) = s_k(x) + s_k(y)$$ 
for all $x,y \in \GF{2^d}$
\end{lemma}

As in \cite{Cantor:1989} \cite{gao2010additive} \cite{Lin:BasisCvt:2016}, we will consider Cantor bases to construct an efficient algorithm.
For the rest of this subsection, we list properties of subspace polynomial with respect to Cantor bases. These properties were proven in \cite{Cantor:1989} and are necessary for deriving the algorithm later.

\begin{lemma}
\label{skvk_one}
For a Cantor basis $\bvec{v}_d$,
$s_k(v_k) = 1$ for $0 \leq k < d$
\end{lemma}

Given a function $f$, denote $f^{\circ i}$ as $\underbrace{\left ( f \circ f \circ \cdots \circ f\right )}_{i\ \text{times}}$, which is function composition $i$ times.
\begin{lemma}
The subspace polynomial with respect to a Cantor basis can be written as a recursive form:
\begin{align*}
s_0(x) &= x\\
s_1(x) &= x^2 +x \\
s_j(x) &= s_{j-1}(x) \cdot s_{j-1}(x-v_{j-1}) = s_{j-1}^2(x) + s_{j-1}(x) \\
&= s_1(s_{j-1}(x)) = s_1^{\circ j}(x)\\
\end{align*}
\end{lemma}

\begin{lemma}
\label{Cantor_property}
Given $d$ a power of two and a Cantor basis $\bvec{v}_d$, then
$$
v_0 = 1
$$
$$
s_1(v_i) = v_i^2 + v_i = v_{i-1} + \alpha
$$
where $\alpha \in W_{i-1}$ for $i>0$.
\end{lemma}

\begin{lemma}
\label{thm:subspace-vanishing-polynomial}
Given a Cantor basis $\bvec{v}_d$, $\forall 0 \leq j \leq k \leq d$.
$$s_j(v_k) + v_{k-j}  \in W_{k-j}$$
\end{lemma}

\begin{lemma}
\label{thm:subspace-vanishing-polynomial-power-of-2}
For subspace polynomial\/ $s_k(x)$ with respect to a Cantor basis $\bvec{v}_d$.
$$
s_k(x) = \sum_{i=0}^k s_{k,i} x^{2^i}
$$
where $s_{k,i} \in \GF{2}$ for $0 \leq i \leq k$.

If k is a power of 2 then
$$
s_k(x) = x^{2^k} + x
$$
\end{lemma}




\subsection{Polynomial basis}
Here we will introduce polynomial basis proposed in \cite{lch-afft} and denote it as \novelpoly. They propose an additive Fast Fourier transform given a polynomial represented with \novelpoly.


\begin{definition}\label{def:novelpoly}
Given a basis $\bvec{v}_d$ and its 
subspace polynomials $(s_0,s_1,...,s_{d-1})$ and $n=2^d$,
define its corresponding \novelpoly basis to be the polynomials $(X_k)_{k=0}^{n - 1}$
$$
X_k(x):= \prod \left(s_i(x)\right)^{b_i} \quad 
\mbox{ where } k=\sum_{i=0}^{d-1} b_i 2^i \mbox{ with } b_i \in \{ 0,1\}\enspace .
$$ and $X_0(x) = 1$.
\end{definition}
Since $\deg(s_i(x)) = 2^i$ for all $i$, $\deg(X_k(x)) =k$ for all $k$.

%

Thus, given any polynomial $P \in \GF{2^d}[x]_{< n}$, it can be represented with \novelpoly,
$$
P = p_0 X_0(x) + p_1 X_1(x) + p_2 X_2(x) + \ldots + p_{n-1} X_{n-1}(x)
$$
where all $p_i \in \GF{2^d}$.

To perform basis conversion between monomial basis and \novelpoly, we can simply recursively divide $s_k(x)$. Thus the cost of naive basis conversion is $O(n (\lg(n))^2)$ additions in $\GF{2^d}$. However, more efficient polynomial basis conversion exists with respect to Cantor basis which was proposed in \cite{gao2010additive} and \cite{Lin:BasisCvt:2016}. 
We show algorithm from \cite{Lin:BasisCvt:2016} in Algorithm \ref{alg:changeBasis2}. 
The algorithm only requires $O(n\lg(n)\lg(\lg(n)))$ additions in $\GF{2^d}$.
It is easy to see that for polynomial admits coefficients $\GF{2}$, basis conversion from monomial to \novelpoly can easily gain a factor of $d$ because addition in $\GF{2}$ cost $A(1)$ instead of $A(d)$.

\begin{algorithm}[!tbh]
\SetKwFunction{cvt}{BasisConversion}
\SetKwFunction{deg}{Degree}
\SetKwInOut{Input}{input}
\SetKwInOut{Output}{output}
\cvt{$f(x)$} : \\
\Input{ $ f(x) = f_0 + f_1 x + ... + f_{n-1} x^{n-1} $ }
\Output{ $ f(x) = g(X) = g_0 + g_1 X_{1}(x) + ... + g_{n-1} X_{n-1}(x) $ }
\BlankLine
\lIf { \deg{$f(x)$} $ \leq 1 $ }
{
  return $g(X) = f_0 + X_1 f_1$
}
Let $k = $ max $\{2^i$ :\deg{ $s_{2^i}(x)$ } $ \leq $ \deg{ $f(x)$ }\} .\\
Let $y = s_k(x)$. \\
Let $f(x) = h'(y) = q'_0(x) + q'_1(x) y + q'_2(x) y^2 + \cdots $ where coefficients of $h'(y)$ are polynomials $q'_i(x)$ whose degree $< 2^{k}$. \\
$h(Y) \gets $ \cvt{ $h'(y)$ } \\
Then we have $h(Y) = q_0(x) + q_1(x) X_{2^k} + q_2(x) X_{2^{k+1}} + \cdots$ \\
$g_i(X)  \gets$ \cvt{ $q_i(x)$} for all $q_i(x)$. \\
\Return {$ g_0(X) + g_1(X) X_{2^k} + g_2(X) X_{2^{k+1}} +\cdots $} \\
\caption{Basis conversion: monomial to \novelpoly constructed from Cantor basis}
\label{alg:changeBasis2}
\end{algorithm}

\subsection{Additive FFT}
Given a polynomial $P$ represented in \novelpoly, Lin, Chung and Han\cite{lch-afft} proposed a fast method to compute its additive Fourier transform.


Given a basis $\bvec{v}_d$ of finite field $\GF{2^d}$, we can construct the polynomial basis accordingly: $(X_0(x),X_1(x),\ldots, X_{2^d-1}(x))$.
Then given a polynomial of $P \in \GF{2^d}[x]_{<2^k}$ represented with \novelpoly $$
P(x) = p_0X_0(x) + p_1X_1(x) + \ldots + p_{2^k-1}X_{2^k-1}(x)
$$, we denote $\lchbtfy(k,P(x),\alpha)$ =
$ (P(\omega_i + \alpha))_{i=0}^{2^k-1}$, where 
$$
P(\omega_i + \alpha) = \sum_{0\leq j < 2^k} p_{j} X_{j}(\omega_{i} + \alpha)
$$ 
$k \leq d$ and $\alpha \in \GF{2^d}$.
Now, let $n_1 = 2^{k-1}$
%
%
%
%
%
%
%
%
%
%
\begin{align*}
&P(\omega_i + \alpha) = P(\omega_{n_1 \cdot i_1 + i_2} + \alpha)&\\
&=  \sum_{0 \leq j_2 < n_1} \sum_{0 \leq j_1 <2} p_{n_1 \cdot j_1 + j_2} X_{n_1 \cdot j_1 + j_2}(\omega_{n_1 \cdot i_1 + i_2} + \alpha) & \\
&= \sum_{0 \leq  j_2 < n_1} \Big( p_{ j_2} + s_{k-1}(\omega_{n_1 \cdot i_1 + i_2 } + \alpha) \cdot p_{n_1+j_2}\Big) X_{j_2}(\omega_{n_1 \cdot i_1 + i_2} + \alpha) & \\
&= \sum_{0 \leq  j_2 < n_1}\Big( p_{ j_2} + s_{k-1}(\omega_{n_1 \cdot i_1 } + \alpha) \cdot p_{n_1+j_2}\Big) X_{j_1}(\omega_{i_2} + (\alpha + \omega_{n_1 \cdot i_1})) &
\end{align*}
We can see that the \lchbtfy with input polynomial degree of $2^k-1$ can be computed using two \lchbtfy with input polynomial
of degree $2^{k-1}-1$ corresponding to $i_1=0$ and $1$. With above derivation, we get the algorithm \ref{alg:fft}.

%
%


\begin{algorithm}[!tbh]
\SetKwFunction{nfft}{\ensuremath{\texttt{AFFT}}}
\SetKwInOut{Input}{input}
\SetKwInOut{Output}{output}
\nfft{$k, P(x), \alpha$} : \\
\Input{ $ P(x) = p_0 X_0(x) + p_1 X_1(x) + ... + p_{2^k-1} X_{2^k-1}(x)$ , all $p_i \in \GF{2^d}$ \\
$\alpha \in \GF{2^d} $, $k \leq d$}
\Output{ $ (P( \omega_0 + \alpha ),P( \omega_1 + \alpha ), \ldots , P( \omega_{2^k-1} + \alpha) ) $ .}
\BlankLine
\lIf { $ k = 0 $ }{ return $ p_0 $ }
\tcp{Decompose $P(x) = P_0(x) + s_{k-1}(x) \cdot P_1(x) $.}
$P_0(x) \gets p_0X_0(x) + p_1X_1(x) + \ldots p_{2^{k-1}-1}X_{2^{k-1}-1}(x)$\\
$P_1(x) \gets p_{2^{k-1}}X_0(x) + p_{2^{k-1}+1}X_1(x) + \ldots p_{2^{k}-1}X_{2^{k-1}-1}(x)$\\
$Q_0(x) \gets P_0(x) + s_{k-1}(\alpha) \cdot P_1(x) $. \\
$Q_1(x) \gets Q_0(x) + s_{k-1}(v_{k-1}) \cdot P_1(x) $. \\

return \nfft{$ k-1,  Q_0(x) , \alpha $}$\|$\nfft{$ k-1, Q_1(x) , v_{k-1} + \alpha$} 
\caption{Addtive FFT in \novelpoly from \cite{lch-afft}}
\label{alg:fft}
\end{algorithm}
Note that if we use Cantor basis, then $s_{k-1}(\omega_{n_1}) =s_{k-1}(v_{k-1}) = 1$ by lemma \ref{skvk_one}.
%
%
Given $P \in \GF{2^d}[x]_{<n}$ represented in monomial basis and $n=2^m$, its additive Fourier transform $\AFT(P)$ can be computed as follow.
We first perform basis conversion to get $p_i$ such that $P(x)=p_0X_0(x) + p_1X_1(x) + \ldots + p_{2^m-1}X_{2^m-1}(x)$. Then we perform $\lchbtfy(m,P(x),0)$.
Thus, to compute \AFT(P) using $\lchbtfy$, the maximum depth of recursion is $m$, and the algorithm performs total $\frac{1}{2} n$ multiplications and $n$ additions in each depth of recursion. Therefore the cost of the algorithm is $\frac{1}{2}n \lg(n) (M(d)+2A(d)) $ where $n=2^m$ is the number of terms.



\section{Frobenius Additive Fourier Transform}
\label{sec:faft}
\subsection{Frobenius additive Fourier transform}
Let $P$ be a polynomial in $\GF{2}[x]$ and $\bvec{v}_d$ be a basis in $\GF{2^d}$. We define the Frobenius map $\phi: x \mapsto x^2$. Notice that
$$P(\phi(a)) = \phi(P(a))$$
for all $a \in \GF{2^d}$.

The core idea of the Frobenius Fourier transform is to evaluate a
minimal number of points and all other points can be computed by
applying Frobenius map $\phi$. This is because we now consider
polynomial in $\GF{2}[x] \subset \GF{2^d}[x]$. 
The set of those points is called a \emph{cross section} \cite{ffft}. Formally, given a
set $W \subseteq \GF{2^d}$, a subset $\Sigma \subseteq W$ is called a
\textit{cross section} of $W$ if for every $w \in W$, there exists
exactly one $\sigma \in \Sigma$ such that
$\phi^{\circ j}(\sigma) = w$ for some $j$.
Let $\bvec{v}_d$ denote a basis of $\GF{2^d}$. Given a polynomial $P \in \GF{2}[x]_{<n}$ where $n=2^m$, then the $\AFT(P)$ is the evaluation of the points in
$W_m = \{\omega_0,\omega_1,\omega_2,\ldots, \omega_{2^m-1}\}$.
To perform Frobenius additive Fourier transform, we partition $W_m$ into disjoint orbits by $\phi$. 
If there exists a subset $\Sigma$ of $W_m$ that contains exactly one element in each orbit, that is, $\Sigma$ is a cross section of $W_m$, then Frobenius mapping allows us to recover $\AFT(P)$ from the evaluations of $P$ at each of the points in $\Sigma$. 
We denote 
$$\{P(\sigma)  | \sigma \in \Sigma\}$$ 
the \textit{Frobenius additive Fourier transform} (\FAFT{}) of polynomial $P$.

To exactly evaluate a polynomial with the points in cross section and reduce the complexity of algorithm by a factor $d$ for discrete Fourier transform is certainly not easy as can be seen in \cite{ffft}.

However, when considering the additive Fourier transform proposed by Cantor, we will show that there exists a cross section such that we can naturally use truncated method (as in truncated FFT) to only evaluate those points. In other words, there exists a cross section suited for the structure of additive FFT and let us obtain a fast algorithm.


\subsection{Frobenuis map and Cantor basis}

Consider the field \GF{2^d} with Cantor basis $\bvec{v}_d$ for
$d$ a power of two.
We have $\phi(v_0) = v_0$, and
$\phi(v_i) = v_i^2 = v_i + v_{i-1} + \alpha$, where
$\alpha\in W_{i-1}$ for $i > 0$ from Lemma~\ref{Cantor_property}.
In this section, we will show how to explicitly construct a cross
section $\Sigma$ for \GF{2^d}.

We recall that the Frobenius map $\phi$ on \GF{2^d} generates the
(cyclic) Galois group \Gal{\GF{2^d}/\GF 2} of order
$[\GF{2^d}:\GF 2]=d$, which naturally acts on \GF{2^d} by taking
$\alpha\in\GF{2^d}$ to $\phi(\alpha)$.
The orbit of $\alpha$ under this action is thus
\[ \Orb\alpha=\left\{\sigma(\alpha):\sigma\in\Gal{\GF{2^d}/\GF
      2}\right\}. \]
\begin{lemma}
  Given a Cantor basis $\bvec{v}_d$, $\forall k > 0$,
  $\forall w\in W_{k+1}\setminus W_k$,
  \[ \big|\Orb w\big| = 2^{\rd k+1}. \]
\end{lemma}
\begin{proof}
  Let $\ell=\rd k$.
  In this case, $2^\ell\leq k<2^{\ell+1}$, and
  $v_k=u_\ell u_{\ell-1}^{j_{\ell-1}}\cdots u_0^{j_0}$,
  $j_i\in\{0,1\}\forall 0\leq i<\ell$.
  Since $w\in W_{k+1}\setminus W_k$, we can write
  \[ w=v_k+\alpha=u_\ell u_{\ell-1}^{j_{\ell-1}}\cdots
    u_0^{j_0}+\alpha \] for some $\alpha\in W_k$.
  Obviously the splitting field of $w$ is
  $\GF{2^{2^{\ell+1}}}=\GF 2(u_0,u_1,\ldots,u_\ell)$, so the
  stabilizer of $w$ is the subgroup of \Gal{\GF{2^d}/\GF 2} generated
  by $\phi^{2^{\ell+1}}$.
  It follows immediately from the orbit-stabilizer theorem and
  Lagrange's theorem that \[ \big|\Orb w\big| = 2^{\rd k+1}. \]
\end{proof}

Moreover, we can further characterize the orbit of
$w\in W_{k+1}\setminus W_k$ using the following lemma.
\begin{lemma}
  \label{thm:orbit}
  Given a Cantor basis $\bvec{v}_d$, $\forall k>0$, consider the orbit
  of $w\in W_{k+1}\setminus W_k$ under the action of \Gal{\GF{2^d}/\GF
    2}.
  Then for all $j_i\in\{0,1\}$, $i=1,2,4,\ldots,\binrd k$, there is
  precisely one element $w'\in\Orb w$ such that
  $w'=v_k+j'_1v_{k-1}+\cdots+j'_kv_0\in W_{k+1}\setminus W_k$,
  $\forall j'_i\in\{0,1\}$, and $j'_i=j_i$ for
  $i=1,2,4,\ldots,\binrd k$.
\end{lemma}
\begin{proof}
  Let $\ell$ be a power of two.
  From Lemma~\ref{thm:subspace-vanishing-polynomial-power-of-2}, we
  have $\phi^{\circ\ell}(x)=x^{2^\ell}=s_\ell(x)+x$.
  From Lemma~\ref{thm:subspace-vanishing-polynomial}, we see that
  $\phi^{\circ \ell}(w)+w=s_\ell(w)\in W_{k-\ell+1}\setminus W_{k-\ell}$.
  That is, $\phi^{\circ \ell}(w)+w\in v_{k-\ell}+W_{k-\ell}$.
  Let $\ell=1$, $\phi^{\circ \ell}$  allows us to obtain $w$ and $\phi(w)$, one of which
  has $j'_1=0$ while other, $j'_1=1$ for any $w \in W_{k+1} \setminus W_k$.
  Now let $\ell=2$.
  We can use $\phi^{\circ \ell}$ to obtain $w$ and $\phi^{\circ 2}(w)$, one of which has $j'_2=0$ while
  the other, $j'_2=1$. Both $w$ and $\phi^{\circ 2}(w)$ have the same $j'_1$. 
  Similarly for $\phi(w)$ and $\phi^{\circ 2}(\phi(w))$.
  Let $\ell=4$. We can use $\phi^{\circ \ell}$ to obtain $w$ and $\phi^{\circ 4}(w)$, one of which has $j'_4=0$ while
  the other, $j'_4=1$ and both $w$ and $\phi^{\circ 4}$ have the same $j'_1,j'_2$.
  If we continue, we can then obtain all combinations of
  $j_i\in\{0,1\}$, for $i=1,2,4,\ldots,\rd k$.
  However, as $|\Orb w|=2^{\rd k+1}$, we see that each such
  combination can appear in \Orb w precisely once due to the
  pigeonhole principle.
\end{proof}

Now we can explicitly construct a cross section.
Let $\Sigma_0 = \{0\}$, and $\forall i>0$, let
\[ \Sigma_i=\left\{v_{i-1} + j_1v_{i-2} + \cdots +
    j_{i-1}v_0 : \begin{aligned}
        & j_k=0 \text{ if $k$ is a power of 2,} \\
        & j_k\in\{0,1\} \text{ otherwise.}
      \end{aligned}\right\} \]
\begin{theorem}
\label{main}
$\Sigma_i$ is a cross section of $W_i \setminus W_{i-1}$.
That is, $\forall i>0$, $\forall w\in W_{i}\setminus W_{i-1}$, there
exists exactly one $\sigma\in\Sigma_i$ such that
$\phi^{\circ j}(\sigma) = w$ for some $j$.
\end{theorem}
\begin{proof}
  First, any two elements of $\Sigma_i$ are in different orbits for
  any $i$; this is a corollary of Lemma~\ref{thm:orbit}.
  Next, we know that $\forall w \in W_{i}\setminus W_{i-1}$,
  $\big|\Orb w\big| = 2^{\binrd{(i-1)}+1}$, and
  $\nexi{j}{w} \in W_{i}\setminus W_{i-1}$, $\forall j$.
  So each orbit generate by element in $\Sigma_{i}$ has the size
  $2^{\binrd{(i-1)}+1}$, and
  $ 2^{\binrd{(i-1)}+1} \cdot |\Sigma_i| = 2^{i-1} = |W_i \setminus
  W_{i-1}|$.
  By the pigeonhole principle, each element in $W_i \setminus W_{i-1}$
  must be in an orbit generate by exactly one element in $\Sigma_i$.
\end{proof}

\subsection{Frobenius additive Fast Fourier transform} 
With the theorem \ref{main}, a cross section of $W_m$ is $$
\Sigma_0 \cup \Sigma_1 \cup \Sigma_2 \cup \ldots \cup \Sigma_{m}
$$
Given $P(x) \in \GF{2}[x]_{<n}$ represented with \novelpoly and Cantor basis $\bvec{v}_d$ of field $\GF{2^d}$ where $n=2^m$, instead of computing $\AFT(P) = (P(\omega_0),P(\omega_1), \ldots, P(\omega_{2^m-1}))$, we only need to compute $\FAFT{n}(P) =  \{P(\sigma) : \sigma \in \Sigma_0 \cup \Sigma_1 \cup \Sigma_2 \cup \ldots \cup \Sigma_{m}\}$ and then use Frobenius map $\phi$ to get the rest.

Due to the structure of the Additive FFT, we can simply `truncate` to those points. 
In the original additive FFT (algorithm \ref{alg:fft}), each \FAFFT calls two \FAFFT routines recursively.
Those two \FAFFT call corresponds to evaluate points in $\alpha + W_{k-1}$  and $\alpha + v_{k-1} + W_{k-1}$.
We can omit one call and only compute $\alpha + W_{k-1}$ so we will not evaluate the points not in the cross section $\Sigma$ when $\Sigma \cap  (\alpha + v_{k-1} + W_{k-1}) = \emptyset$. 
Then we get the algorithm \ref{alg:fafft}.

It is easy to see that $\FAFFT(m,P(x),1,v_m)$ computes $\{P(x): x \in \Sigma_m\}$
because truncation happens when the $v_{m-l}$ component is zero for all points in $\Sigma_m$ and $l$ is a power of two. 
To compute \FAFT{}(P), we call $\FAFFT(m,P(x),0,0)$.
\begin{algorithm}[!tbh]
\SetKwFunction{nfft}{\ensuremath{\texttt{FAFFT}}}
\SetKwInOut{Input}{input}
\SetKwInOut{Output}{output} 
\nfft{$k, P(x), l, \alpha$} : \\
\Input{ 
        $k \in \mathbb{N} \cup \{0\}$.
        $ P(x) = p_0 X_0(x) + p_1 X_1(x) + ... + p_{2^k-1} X_{2^k-1}(x)$ : $p_i \in \GF{  2^{\binru{l}} }$ if $l> 0$, $p_i \in \GF{2}$ otherwise.\\
        $l \in \mathbb{N} \cup \{0\}$\\
        $ \alpha \in W_{k+l}\setminus W_{k} $ if $l>0$ , otherwise $\alpha = 0$.
}
\Output{
$P(\sigma)_{\sigma \in \Sigma}$ where 
$\Sigma  = (\Sigma_{0} \cup \Sigma_{1} \cup \ldots \cup \Sigma_{k+l})\cap (\alpha + W_{k}) $, 
}
\BlankLine
\lIf { $ k = 0 $ }{ \Return $ p_0 $ }
Decompose $P(x) = P_0(x) + s_{k-1}(x) \cdot P_1(x) $. \\
$Q_0(x) \gets P_0(x) + s_{k-1}(\alpha) \cdot P_1(x) $. \\
$Q_1(x) \gets Q_0(x) +  P_1(x) $. \\
\uIf { $l = 0$ }{
    \Return \nfft{$ k-1,  Q_0(x) , l, \alpha$} $\|$
    \nfft{$ k-1,  Q_1(x) , l+1, v_{k-1} + \alpha$}\\
}
\uElseIf {$l$ is a power of two}{
    \Return \nfft{$ k-1,  Q_0(x), l+1, \alpha$}\\
}\Else{
 \Return {
 \nfft{$ k-1,  Q_0(x) , l+1 ,\alpha$} $\|$ \nfft{$ k-1,  Q_1(x) , l+1, v_{k-1} + \alpha$}
 }
}
\caption{Frobenius Additive FFT in \novelpoly.}
\label{alg:fafft}
\end{algorithm}

The Fig. \ref{fig:butterfly} is a graphical illustration of $\FAFFT(5,f,0,0)$ routine which computes $\FAFT{32}(f)$ where $f=g_0X_0(x)+g_1X_1(x)+g_2 X_2(x) + \ldots +g_{31} X_{31}(x)$. 
It consists of 5 layers corresponding to the recursive depth in the pseudocode. 
Each grey box is a `butterfly unit` that performs a multiplication and an addition.
A butterfly unit has two inputs $a,b \in \GF{2^d}$. For normal butterfly unit with two output $a',b'$, it performs
\begin{align*}
a' &\gets a+b\cdot s_k(\alpha)\\
b' &\gets a'+b
\end{align*}
while the truncated one only output $a'$.
In the figure, we denote the $s_k(\alpha)$ in each butterfly unit $c_{i,j}$.
Initially, the input of butterfly unit, $g_0,g_1,\ldots, g_{31}$, are all in $\GF{2}$. 
But as it goes through layer by layer, because the multiplicands $c_{i,j}$ maybe in extension fields,
the bit size of input to the following butterfly unit grows larger.
For example, after second layer, the lower half of the input are in $\GF{2^2}$ because $c_{3,1}$ are in $(W_2 \setminus W_1) \subset \GF{2^2}$.
Then they go through butterfly unit with $c_{2,2} \in (W_3\setminus W_2) \subset \GF{2^4}$ and come to be in $\GF{2^4}$.
\begin{figure}[ht]
	\caption{Illustration of the butterfly network with $n=32$. 
	$f(0), f(1) \in \GF{2}, 	f(\omega_2) \in \GF{2^2}, 
	f(\omega_4),f(\omega_8), f(\omega_9) \in \GF{2^4}$ and
	$f(\omega_{16}),f(\omega_{18}) \in \GF{2^8}$}
	\label{fig:butterfly}
	\includegraphics[width=8.5cm]{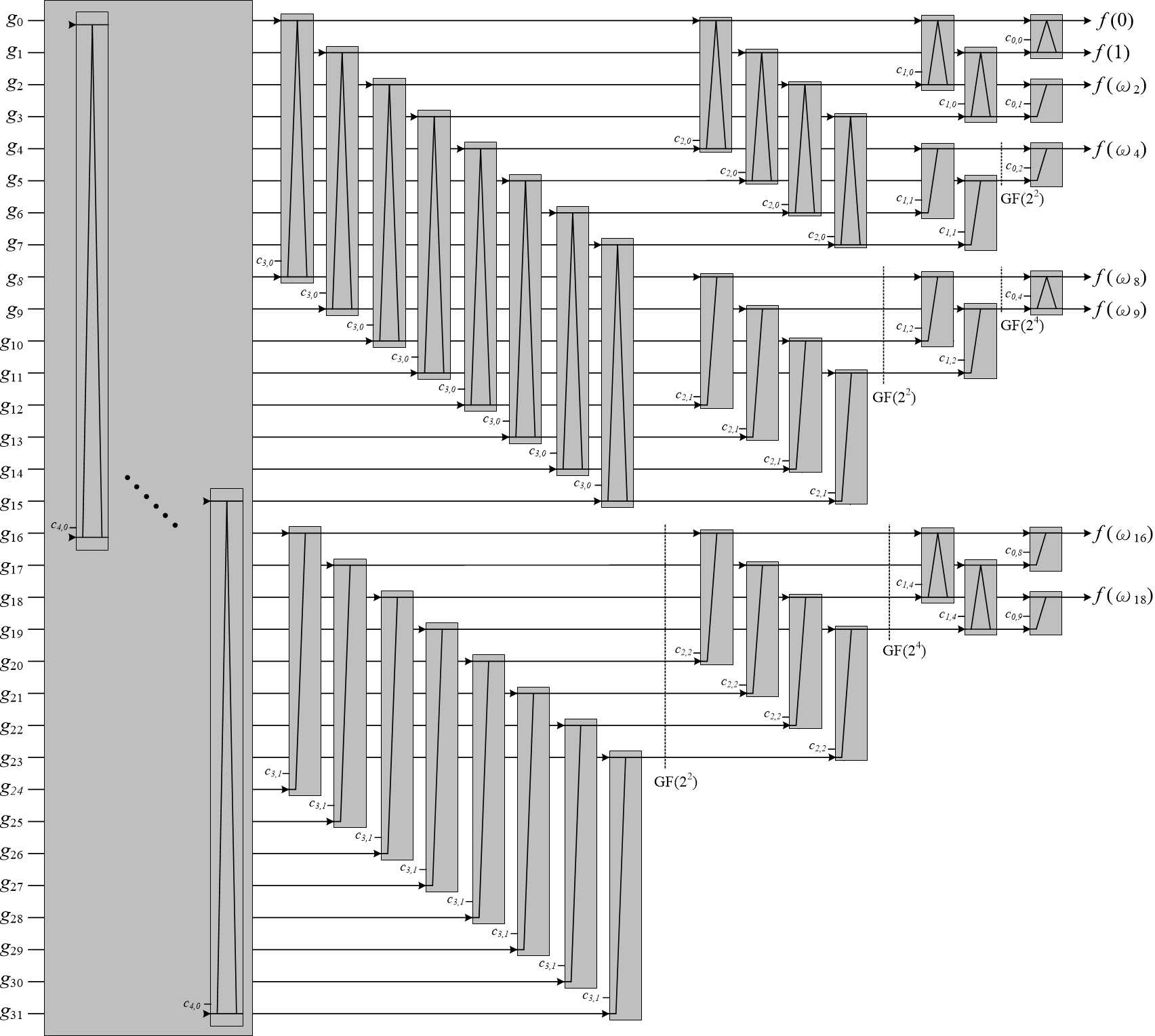}
\end{figure}

\subsection{Complexity Analysis}
In this section, we analyze the complexity of $\FAFFT$ in algorithm \ref{alg:fafft}.
Let $F(k,l)$ and $F_A(k,l)$ denote the cost of multiplication and addition to compute $\FAFFT(k,P(x),l,\alpha)$ 
for $P(x) \in \GF{2}[x]_{<2^k}$ and $\alpha \in W_{k+l}\setminus W_k$.

First, it is straightforward to verify that for all $\FAFFT(k',P'(x),l',\alpha')$
call during recursion:
\begin{itemize}
    \item  $\alpha' \in W_{k'+l'}\setminus W_{k'} $ if $l'>0$ , otherwise $\alpha' = 0$ 
    \item  $P'(x)=\sum p'_iX_i(x)$, $p'_i \in \GF{  2^{\binru{l'}} }$ if $l'> 0$, $p'_i \in \GF{2}$ otherwise.
    \item $ \!
      \begin{aligned}[t]
      s_{k-1}(\alpha') & \in (W_{l'+1}\setminus W_{l'}) \\
      &\in
      \begin{cases}
       u_{\lg{l'}} + \GF{l'} & \text{if } l' \text{ is a power of two}\\
       \GF{\binru{l'}} & \text{ otherwise}
      \end{cases}
      \end{aligned}
      $
\end{itemize}
Then we have
\[
    F(k,l)= 
\begin{cases}
    F(k-1,l)+F(k-1,l+1)+2^{k-1}(M(1)) & \text{if } l = 0\\
    F(k-1,l+1)+ 2^{k-1} (M(l))             & \text{if } l \text{ is a power of two}\\
    2 \cdot F(k-1,l+1)+ 2^{k-1}( M(\binru{l})) & \text{otherwise}
\end{cases}
\]

\begin{theorem}
\textnormal{(multiplication complexity)} Given $n = 2^m$, for $m+l \leq d$, d is a power of two.
Then we have
\[
    F(m ,l)\leq
\begin{cases}
    \frac{1}{2}(n \lg n \frac{M(d)}{d} ) & \text{if } l = 0\\
    \frac{1}{2}(n \lg n \frac{M(d)}{d}\binru{l})            & \text{otherwise}\\
\end{cases}
\]
\end{theorem}
\begin{proof}
We prove by induction. Consider $m=1$,
then $F(1,l) = M(l) \leq \frac{M(d)}{d} l$ is correct.

Assume $m=k-1$ and for any $l \leq d-m$, 
\[
    F(m,l) \leq 
\begin{cases}
    \frac{1}{2} m 2^m \frac{M(d)}{d}  & \text{if } l = 0\\
    \frac{1}{2} m 2^m \frac{M(d)}{d}\binru{l}          & \text{otherwise}\\
\end{cases}
\]
Then we check three cases:
first, $m=k$ and $l=0$:
\begin{align*}
F(k,l) &= F(k-1,0)+ F(k-1,1) +  2^{k-1} \cdot M(1) \\
       &=  \frac{1}{2}(k-1) 2^{k} \frac{M(d)}{d} +  2^k \cdot M(1) \\
       &\leq \frac{1}{2} k 2^k \frac{M(d)}{d}
\end{align*}
Second, $m=k$ and 
$l$ is a power of two: 
\begin{align*}
F(k,l) &= F(k-1,l+1)+  \cdot 2^{k-1} \cdot M(l) \\
       &=  (k-1) 2^k \frac{M(d)}{d} l +  2^{k-1} \cdot M(l) \\
       &\leq \frac{1}{2} (k-1) 2^k \frac{M(d)}{d} l +  2^{k-1} \frac{ M(d)}{d}l \\
       &= \frac{1}{2} k 2^k \frac{M(d)}{d} l
\end{align*}
Finally, $l>0$ and is not a power of two:
\begin{align*}
F(k,l) &= 2 \cdot F(k-1,l+1)+  2^{k-1} \cdot M(\binru{l}) \\
       &= \frac{1}{2} (k-1) 2^k \frac{M(d)}{d} \binru{l+1} +  2^{k-1} \cdot M(\binru{l}) \\
       &\leq \frac{1}{2} (k-1) 2^k \frac{M(d)}{d} \binru{l} + 2^{k-1} \cdot \frac{M(\binru{l})}{\binru{l}} \binru{l} \\
       &\leq \frac{1}{2} k 2^k \frac{M(d)}{d} l
\end{align*}
Note that we 
assume $\frac{M(l)}{l}$ is increasing in $l$. We complete the proof.

\end{proof}

For the cost of addition, it can be proved follow the same procedure above since each with $2A(d)$ instead of $M(d)$. (Note that $\frac{A(d)}{d}$ is constant)

\begin{theorem}
\textnormal{(addition complexity)} Given $n = 2^m$, for $m+l \leq d$, d is a power of two.
Then we have
\[
    F_A(m,l)\leq
\begin{cases}
    (n \lg n \frac{A(d)}{d} ) & \text{if } l = 0\\
    (n \lg n \frac{A(d)}{d}\binru{l})            & \text{otherwise}\\
\end{cases}
\]
\end{theorem}

Given $P \in \GF{2}[x]_{<n}$ and $n$ a power of two, to compute $\FAFT{n}(P)$,
we call $\FAFFT(\lg(n), P, 0, 0)$.
Thus, the cost of compute $\FAFT{n}(P)$ is 
$\frac{1}{2} n\lg{(n)} \frac{M(d)}{d} + n\lg{(n)} \frac{A(d)}{d}$.
Compare with the additive FFT for $\GF{2^d}[x]_{<n}$ whose cost is $\frac{1}{2} (n \lg(n) (M(d)+2A(d))$,
we gain a speed-up factor $d$.




\subsection{Inverse Frobenius additive FFT}
The inverse Frobenius additive FFT is straight forward because for the butterfly unit with two output, it is easy to find its inverse. 

However, due to the truncation, it is not obvious how to inverse when $l$ is a power of two. Here we show that it is always invertible.
In the algorithm \ref{alg:fafft}, when $l$ is a power of two, it truncates and only compute \FAFT{} ~of $Q_0(x) = P_0(x) + s_{k-1}(\alpha)\cdot P_1(x)$.
To be able to inverse, we need to recover $P_0(x)$ and $P_1(x)$ from $Q_0(x)$. Note that  $s_{k-1}(\alpha) \in (W_{l+1}\setminus W_{l}) = v_{l}+W_{l}$ because $\alpha \in W_{k+l+1} \in W_{k+l}$ and lemma \ref{thm:subspace-vanishing-polynomial}.
Since we use Cantor basis, recall the definition \ref{Cantor_basis}, $v_{l} = u_{\lg{l}}$  when $l$ is a power of two.
We can rewrite the equation from the point of $\GF{2^l}[u_{\lg{l}}][x]$.
Let $s_{k-1}(\alpha) = u_{\lg l} + c$ and $c \in \GF{2^l}$, 
$$Q_0(x) = R_0(x) + R_1(x) u_{\lg l} = P_0(x) + (c+u_{\lg l})\cdot P_1(x)$$ where $R_0(x), R_1(x) \in \GF{2^l}[x]$.
Then we get 
\begin{align*}
    P_0(x) &= R_0(x) + R_1(x) \cdot c\\
    P_1(x) &= R_1(x) 
\end{align*}
Thus we can always recover $P_0(x)$ and $P_1(x)$ from $Q(x)$.
The full inverse Frobenius additive FFT algorithm is shown in algorithm \ref{alg:ifafft}.

\begin{algorithm}[!tbh]
\SetKwFunction{ifft}{\ensuremath{\texttt{IFAFFT}}}
\SetKwInOut{Input}{input}
\SetKwInOut{Output}{output} 
\ifft{$k, A, l, \alpha$} : \\
\Input{$A = P(\sigma)_{\sigma \in \Sigma}$ where $\Sigma  = (\Sigma_{0} \cup \Sigma_{1} \cup \ldots \cup \Sigma_{k+l})\cap (\alpha + W_{k}) $,\\
    $ \alpha \in W_{k+l}\setminus W_{k} $ if $l>0$ , otherwise $\alpha = 0$.
}
\Output{
        $ P(x) = p_0 X_0(x) + p_1 X_1(x) + ... + p_{2^k-1} X_{2^k-1}(x)$ : $p_i \in \GF{  2^{\binru{l}} }$ if $l> 0$, $p_i \in \GF{2}$ otherwise.\\
}
\BlankLine
\lIf { $ k = 0 $ }{ \Return  the only element in $A$ }
\uIf{$l=0$}{
Divide the set $A$ to $A_0$, $A_1$ \\
$Q_0(x) \gets$ \ifft{$ k-1,  A_0 , l, \alpha$} \\
$Q_1(x) \gets$ \ifft{$ k-1,  A_1 , l+1, v_{k-1} +\alpha$}\\
$P_1(x) \gets (Q_0(x) + Q_1(x)) $\\
$P_0(x) \gets Q_0(x) + s_{k-1}(\alpha)\cdot P_1(x)$\\
}
\uElseIf {$l = \binrd{(l)}$}{
    $Q(x) \gets $\ifft{$ k-1, A, l+1, \alpha$}\\
    Let $s_{k-1}(\alpha) = c + u_{\lg(l)}$\\
    Let $Q(x) = R_0(x) + u_{\lg(l)}\cdot R_1(x)$\\
    $P_0(x) \gets R_0(x) + R_1(x) \cdot c$ \\
    $P_1(x) \gets R_1(x)$ \\
}
\Else{
Divide the set $A$ to $A_0$, $A_1$ \\
$Q_0(x) \gets$ \ifft{$ k-1,  A_0 , l+1, \alpha$} \\
$Q_1(x) \gets$ \ifft{$ k-1,  A_1 , l+1, v_{k-1} + \alpha$}\\
$P_1(x) \gets Q_0(x) + Q_1(x)$ \\
$P_0(x) \gets Q_0(x) + s_{k-1}(\alpha)\cdot P_1(x)$\\
}
\Return {
  $P_0(x) + P_1(x) \cdot s_{k-1}(x)$
}
\caption{Inverse Frobenius Additive FFT in \novelpoly.}
\label{alg:ifafft}
\end{algorithm}



\section{Multiplications in $\GF{2}[x]$}
\label{sec:application}

To multiply a polynomial using Frobenius additive FFT is exactly the same as
conventional way: applying basis conversion to convert to \novelpoly, computing Frobenius additive FFT,
pair-wise multiplication, computing the inverse Frobenius additive FFT, then transforming back into
the original monomial basis.

\subsection{Multiplications of $\GF{2}[x]$ of small degree}
%

%
%
%
The record of minimal bit-operation to multiply polynomial over $\GF{2}[x]$ was set by 
\cite{DBLP:conf/crypto/Bernstein09} and \cite{DBLP:journals/jce/CenkH15}, which are both
based on Karatsuba-like algorithm.
Instead of Karatsuba-like algorithm, we use Frobenius additive FFT to perform multiplication in $\GF{2}[x]$. We implement a generator to generate code of binary polynomial multiplication with size $2^m$ where each variable is in $\GF{2}$. Since the multiplicands $s_k(\alpha)$ in FAFFT can all be precomputed. To reduce the number of bit operations, when multiplying a constant, we transform it into a matrix vector product over $\GF{2}$ and apply common subexpression algorithm as in \cite{paar}.
The generated code consists of XOR and AND expressions. The generator will be made public available on Github.

In figure \ref{fig:compareBitPoly}, we show the best results of polynomial multiplication over binary field.
\cite{DBLP:conf/crypto/Bernstein09} set the record for polynomial size up to 1000 in 2009. 
\cite{DBLP:journals/jce/CenkH15} improve the results up to 4.5\% for certain size of polynomial.
Since our Frobenius Additive FFT works with the polynomial size equal to power of two, we apply it to polynomial multiplication with polynomial size 256, 512, and 1024. We improve the best known results by
19.1\%, 29.7\%, and 41.1\% respectively.
To conclude the comparison, we set the record of size 231 to 256, 414 to 512, and 709 to 1024 just by above result.
To the best of our knowledge, it is the first time FFT-based method outperforms Karatsuba-like algorithm in such low
degree in terms of bit operation count.
In addition, for polynomial size 128, our FAFFT costs 11556 bit operations, comparing to 
the best previous is 11466 from  \cite{DBLP:journals/jce/CenkH15}. Our result is only 0.78\% slight slower in terms of bit operation count.

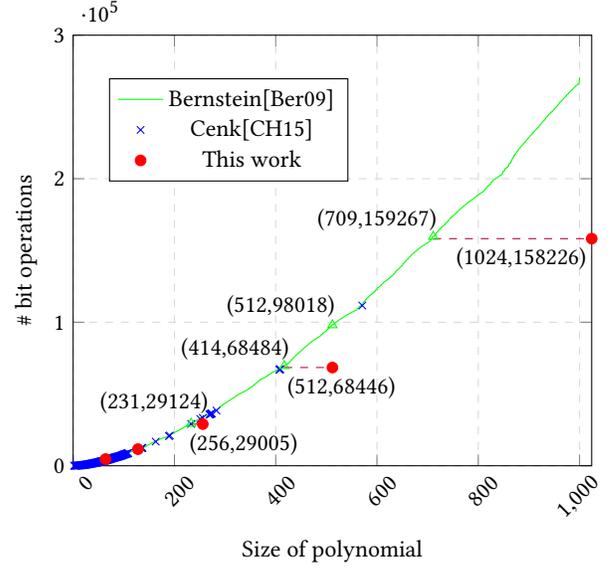
\begin{figure}[th!]
  \begin{center}
    \begin{tikzpicture}
      \begin{axis}[
          width=\linewidth, 
          grid=major, 
          grid style={dashed,gray!30}, 
          xmin=0, xmax=1024,
          ymin=0, ymax=300000,
          xlabel= Size of polynomial, 
          ylabel=\# bit operations,
          legend style={at={(0.3,0.90)},anchor=north}, 
          xticklabel style = {xshift=0.3cm, yshift=-0.1cm},
          x tick label style={rotate=45,anchor=east} 
        ]
        \addplot [color=green,]
        table[x=n,y=djb,col sep=comma] {bitOps.csv}; 
        
        \addplot [only marks,color=blue,mark=x,]
        table[x=n0,y=cenk,col sep=comma] {bitOps.csv}; 
        
        \addplot [only marks,color=red,mark=*,]
        table[x=n,y=wd,col sep=comma] {bitOps.csv}; 
        \addplot [only marks,color=green,mark=triangle,]coordinates{(512,98018)};
        
        \addplot [domain=711:1024,samples=100,color=purple,dashed]{158226};
        \addplot [domain=417:512,samples=100,color=purple,dashed]{68446};
        \addplot [domain=233:256,samples=100,color=purple,dashed]{29005};

        \legend{Bernstein\cite{DBLP:conf/crypto/Bernstein09}, Cenk\cite{DBLP:journals/jce/CenkH15}, This work}
        \addplot[only marks, color=green,mark=triangle,]
coordinates {(711,159622)(417,69574)(233,29354)};
        \node [below, yshift=0ex] at (axis cs:600,185000) {(709,159267)};
        \node [below, yshift=0ex] at (axis cs:320,95000) {(414,68484)};
        \node [below, yshift=0ex] at (axis cs:160,58000) {(231,29124)};
        \node [below, xshift=-6.9ex] at (axis cs:1024,158226) {(1024,158226)};
        \node [below, xshift=0ex] at (axis cs:530,68446) {(512,68446)};
        \node [below, xshift=0ex] at (axis cs:410,125018) {(512,98018)};
        \node [below, xshift=4ex] at (axis cs:256,29005) {(256,29005)};
      \end{axis}
    \end{tikzpicture}
    \caption{Number of bit operations for multiplication of $\GF{2}[x]$}
      \label{fig:compareBitPoly}

  \end{center}
\end{figure}

\subsubsection{Other FFT-based multiplication using Kronecker method}
In \cite{Auth256}, an optimized implementation of additive FFT
based on \cite{gao2010additive} was presented.
They show the cost of multiplication in $\GF{2^8}[x]_{<32}$ is 22,292 bit operations. 
We can use it to multiply polynomials of degree $128$ using Kronecker method. 
But our Frobenius additive FFT only requires 11556 for size 128, which is about
half of their results. The factor $2$ speedup compared with Kronecker method is expected as in \cite{ffft} because the total bit length is half when using Frobenius method.

\subsubsection{Application to Binary Elliptic Curve Cryptography}
There are several polynomial sizes that are in the interest of cryptography engineering community and its number of bit operation of multiplication were studied due to its application in binary elliptic curve\cite{DBLP:conf/crypto/Bernstein09} \cite{DBLP:journals/jce/CenkH15}.
These binary elliptic curve includes:
Koblitz curve \texttt{sect233k1}, \texttt{sect233r1} over $\GF{2}[x]/(x^{233} + x^{74} + 1)$,
curve \texttt{sect239k1}, \texttt{sect239r1} over $\GF{2}[x]/(x^{239} + x^{158} + 1)$,
and Edwards curve \texttt{BBE251} over $\GF{2}[x]/(x^{251} + x^7 + x^4 + x^2 + 1)$ according to
Standards for Efficient Cryptography Group (SECG) and \cite{DBLP:conf/crypto/Bernstein09}.
For the corresponding polynomial size $233$, $239$ and $251$, the Frobenius additive FFT method outperforms previous method in terms of number of bit operations.
Thus, our FAFFT can potentially applied to these curve in order to accelerate the computation.

 \subsection{Multiplication of $\GF{2}[x]$ of large degree}
 Another application is to implement multiplications of $\GF{2}[x]$ of large degree on modern CPU. Here we will implement a variant of the algorithm. The bit operation count is not a good predictor on modern CPU since they operate on 64-bit machine words and there are special instruction \texttt{PCLMULQDQ} designed for carryless multiplication with input size $64$. As in \cite{vanderhoeven:hal-01579863}, to implement on modern CPU, we have to take these into account. 
 To be able to use \texttt{PCLMULQDQ}, we change our algorithm to only compute a subset of cross section. The set of point we will use is 
 \begin{align*}
\{v_{i-1} + v_{i-2} j_1 + v_{i-3} j_2 + \ldots + v_{0} j_{i-1} : &j_k=0 \text{ if }k \leq 64,\\
 &j_k \in \{0,1\} \text{ otherwise}   \} \\
\end{align*}
where $ 64 < i\leq 128$.
By selecting this subset, we can mostly operate in $\GF{2^{128}}$ and mainly use the \texttt{PCLMULQDQ} instruction which performs carryless multiplication with input size 64.
We show the benchmark on Intel Skylake architecture in Table~\ref{tab:angel1} with comparison of other implementations.
In the table, for polynomial of size $n$ where $\log_2 (n/64) = 16,17,\ldots, 23$, our implementation of variant of FAFFT outperforms previous best results from \cite{vanderhoeven:hal-01579863,DBLP:journals/corr/abs-1708-09746,DBLP:conf/issac/HarveyHL16,gf2x}.

\begin{table}
\begin{center}
\begin{threeparttable}
\caption{Products in degree $< n$ in $\GF{2}[x]$ on Intel Skylake Xeon E3-1275 v5 @ 3.60GHz ($10^{-3}$ sec.)
}
\label{tab:angel1}
\begin{tabular}{l    r  r r  r  r  r  r  r  }
\toprule
   $\log_2 (n/64)$  & 16  & 17 &  18 & 19 & 20 & 21 & 22 & 23   \\ 
\midrule
This work, $\GF{2^{128}}$  & 9 & 20 &  41 & 88 & 192 & 418 & 889 & 1865 \\
FDFT \cite{vanderhoeven:hal-01579863} \tnote{c}  &  11 &  24  &  56   & 127   & 239   &  574  & 958 & 2465 \\ 
ADFT\cite{DBLP:journals/corr/abs-1708-09746}    & 16 & 34 &  74 & 175 & 382 & 817 & 1734 & 3666  \\ 
$\GF{2^{60}}$\cite{DBLP:conf/issac/HarveyHL16}\tnote{b}   &  22 &  51  &  116   & 217   & 533   &  885  & 2286 & 5301 \\ 
\texttt{gf2x} \tnote{a} \cite{gf2x}  & 23 & 51 & 111 & 250 & 507 & 1182 & 2614 & 6195  \\ 
\bottomrule
\hline
\end{tabular}
\begin{tablenotes}
  \small
  \item [a] Version 1.2.  Available from \url{http://gf2x.gforge.inria.fr/}
  \item [b] SVN r10663. Available from \url{svn://scm.gforge.inria.fr/svn/mmx}
  \item [c] SVN r10681. Available from \url{svn://scm.gforge.inria.fr/svn/mmx}
\end{tablenotes}
\end{threeparttable}
\end{center}
\end{table}




\section{Future Direction}
This is the first time FFT-based algorithm that outperforms Karatsuba-like
 algorithm for binary polynomial multiplication in such low degree.
 We hope our work can open up a new direction for the community interested in the number 
 bit operation of binary polynomial multiplication in small degree, as there are possible future work such as further reducing bit operations in \cite{Auth256} or using truncated method to eliminate the `jump` in the complexity when size is a power of two \cite{DBLP:conf/issac/Hoeven04}.





\bibliographystyle{alpha}
\bibliography{scientific,superstition} 

\end{document}
